\documentclass[12pt,reqno]{amsart}
\usepackage[utf8]{inputenc}
\usepackage{amsmath,amssymb,amsthm,hyperref,array,xcolor,mathtools,multicol,subcaption,verbatim,microtype}
\hypersetup{
pdftitle={Mechanism Design by a Politician}, 
pdfauthor={Giovanni Valvassori Bolgè}, 
pdfkeywords={Legislative Coalitions, Public Goods, Mechanism Design, Reservation Utility}, 
pdfcreator={LaTeX with hyperref package}, 
pdfproducer={LaTeX},
colorlinks=true,
linkcolor=[rgb]{0,0,0.6}, 
citecolor=[rgb]{0,0,0.6},  
urlcolor=[rgb]{0,0,0.6},   
}

\usepackage[normalem]{ulem}
\usepackage[pdftex]{graphicx}
\usepackage{fullpage}
\usepackage{cleveref}

\usepackage{tikz}
\usepackage{tikz-3dplot}
\usetikzlibrary{patterns,patterns.meta,arrows.meta, calc}

\usepackage{enumitem}
\setlist[itemize]{nosep}
\setlist[enumerate]{nosep}

\usepackage{natbib}

\newtheorem*{theorem*}{Theorem}
\newtheorem{theorem}{Theorem}
\newtheorem{lemma}{Lemma} 
 
\newtheorem{definition}{Definition} 
\newtheorem{corollary}{Corollary} 
\newtheorem{assumption}{Assumption} 
\newtheorem{proposition}{Proposition} 

\newtheorem{step}{Step}

\theoremstyle{remark}

\DeclareMathOperator*{\argmax}{arg\,max}

\DeclareMathOperator{\ext}{ext}

\newcommand{\Mprob}{\Delta(\X)}
\newcommand{\F}{\mathcal{F}}

\newcommand{\X}{\mathcal{X}}

\usepackage{accents}

\crefname{equation}{}{}

\begin{document}
	
\title{Optimal Large-Contest Design*}
\thanks{*I am particularly indebted to Ioan Manolescu for his generous time. I also thank Yaron Azrieli, Christian Chambers, Christian Ewerhart, Dominik Inauen, Byeong-hyeon Jeong, Patrick Lahr, Igor Letina, Brendan Pass, Aner Sela, Philipp Strack and Stefan Wenger, and participants at BSE Summer Forum, SGVS, and the Conference on Mechanism and Institution Design for helpful discussions.}
	 
\author{Giovanni Valvassori Bolgè\textsuperscript{\textdagger}}\thanks{\textsuperscript{\textdagger}University of Fribourg, Department of Economics. \textit{Email: \href{mailto:giovanni.valvassoribolge@unifr.ch}{giovanni.valvassoribolge@unifr.ch}}}
	
\date{\today}

\begin{abstract} 

    In this paper, we provide a theoretical framework to analyze the optimal design of large contests, such as university admissions and research funding. Under a broad class of distributional preferences, the optimal allocation rules are step functions with at most two steps, thus rationalizing several real-world examples. From a methodological perspective, we leverage the mathematics of extreme points to characterize optimal allocation rules in environments where the distribution of efforts can be endogenous. When the designer wants to maximize aggregate effort, the optimal allocation rule sets a performance threshold and it awards prizes at random above it, provided noise is sufficiently small. When noise is larger, awarding prizes to the best performers becomes optimal.

	
\end{abstract}

\maketitle

\newpage

\section{Introduction}
The theory of contests has provided a unifying framework to model a wide range of economic environments.\footnote{For a thorough review, see \cite{konrad_strategy_2009}} Although the classic contributions consider a setting with a finite number of agents\footnote{See, \textit{inter alia}, \cite{lazear_rank-order_1981} and \cite{tullock_efficient_2001}.}, some recent contributions analyze contests with a continuum of agents.\footnote{See \cite{olszewski_large_2016}, \cite{adda_grantmaking_2024} and \cite{azrieli_success_2024} among others.} 

In this paper, we consider the optimal large-contest design for a broad class of distributional preferences, and we are interested in characterizing the optimal allocation rules, namely the mapping from the (possibly noisy) performance of the agents to the prizes. 

The design of allocation rules raises fundamental questions in terms of efficiency and equity, particularly so when the stakes of the competition are extremely high. Consider for example an early-career researcher who applies for funding schemes. Alternatively, consider a student who applies to admission to a selective university, or a medical school. These competitions involve a significant amount of resources and their funding decisions will fundamentally shape the successful applicants' careers.\footnote{The latter application is particularly salient due to the extraordinary increase in the submissions for the leading research funding scheme in the European Union in Fall 2026. According to the European Research Council, the number of applications for the latest Starting Grant call went up 60.4\% with respect to last year. Moreover, applications for Advanced Grants increased by 37\% between 2024 and 2026. See \textcolor{blue}{https://www.brusselstimes.com/2380431/eu-research-grant-applications-surge-60-amid-e594m-funding-limit}.}

Therefore, from the designer's perspective, the optimal design of the allocation rule is instrumental in achieving the desired efficiency or redistributive goal.

In this paper, we show that, for a broad class of distributional preferences of the contest designer, the best the designer can to is to set a two-tier prize scheme the optimal allocation rule is a step function with at most two steps. This result provides an economic foundation to several real-world examples in contests for university admission and research funding,

For example, the Swiss National Science Foundation (SNSF) evaluates funding decisions for grants by selecting a threshold performance (`funding line');  grant proposals that rank above the funding line will obtain funding, whereas those below the funding line will not receive any. Moreover, if multiple proposals place themselves at the funding line, a lottery might be involved to select those who will receive funding.\footnote{Details can be found at \textcolor{blue}{https://www.snf.ch/en/6cs2wnfJtcfFDL6o/page/evaluation-procedure}.}

Recently, there has been an increased adoption of lotteries by funding bodies in order to address concerns about the ranking of proposals of seemingly indistinguishable quality. For example, The British Academy in the UK has introduced a randomization in the evaluation of a class of grants for all proposals meeting a quality threshold.\footnote{See \textcolor{blue}{https://www.nature.com/articles/d41586-026-02082-7}.}

In the context of university admission, the adoption of lotteries is seen as a tool to promote equal opportunity and improve diversity among students. In some European countries, such as the Netherlands and Germany, weighted or unweighted lotteries determine the allocation of some places which remain after a first round of admission in medical school.\footnote{For example, the University of Jena recently introduced a lottery (\textit{`losverfahren'}) to assign those places that are left out after a first round of selection; see \textcolor{blue}{https://jenamedia.de/uni-jena-extra-chance-auf-studienplatz-per-losverfahren}. At the University of Amsterdam, the admission cycle for the academic year 2026/2027 in medicine will take the form of an unweighted lottery. According to the Vice Dean of Education at the University of Amsterdam: ``Our aim is to promote equal opportunities and diversity among potential students. After all, it gives all candidates an equal chance of being admitted''. See \textcolor{blue}{https://student.uva.nl/artikelen/2025-bachelor-geneeskunde-uva-kiest-voor-loting-bij-selectie-van-eerstejaars-2}.}

This paper provides a theoretical framework to analyze the optimality of such prize schemes by asking the following question: \textit{for which preferences of the designer are these prize schemes optimal?}

In order to derive the optimal contest design, we leverage the geometric properties of the set of admissible allocation rules. In particular, we provide a novel characterization of the extreme points of a particular class of multidimensional monotone functions which naturally arises in large contests.

For the design of \textit{robust contests}, namely contests in which the distribution of efforts is exogenous, this characterization combined with techniques from \cite{kleiner_extreme_2021} is enough to derive the optimality of two-tier allocation rules, or step functions with at most two steps. This characterization holds in contests with both complete and incomplete information.

When agents strategically optimize their effort levels depending on the allocation rule they face, the distribution of effort becomes endogenous to the allocation rule chosen, and the extreme point techniques in the space of allocation rules are no longer applicable. Therefore, we define a dual problem on the space of effort distributions in which similar techniques can be employed to derive the optimal allocation rules and the equilibrium distribution of efforts.

Interestingly, when the designer wants to maximize aggregate effort, the optimal allocation rule consists in setting a performance threshold and it awards prizes \textit{at random} above it, when effort is perfectly observable. This characterization provides a foundation for the prize scheme devised by the British Academy, but for a different rationale: when the quality of proposals is perfectly observable, introducing a randomization increases aggregate effort as it induces more participation from agents who might have stayed out under a more `meritocratic' prize scheme.

Strikingly, this result goes through even when performance is noisy, provided the noise is sufficiently small. When noise is larger, a fully meritocratic allocation rule is optimal: awarding the prizes to the best performers is the only way to ensure high effort from the best performers, even though this will result in lower participation rates.

The paper is structured as follows. Section \ref{sec:extreme} introduces the mathematical framework and derives the extreme points result. Section \ref{sec:applications} then applies the result to large contests. Section \ref{sec:conclusions} concludes.

\subsection*{Related Literature} This paper contributes to the theory of large contests and to a growing literature on the characterization of extreme points in economics.

Following the pioneering contributions by \cite{olszewski_large_2016} and \cite{olszewski_large_2016}, \cite{morgan_limits_2022} and \cite{azrieli_success_2024} have considered contests with a continuum of agents and provided key insights into how this setting significantly simplifies equilibrium analysis with respect to the classic contest models with a finite number of agents. In particular, \cite{azrieli_success_2024} provide an axiomatization of a class of Contest Success Functions (CSFs) in large contests. They show that any CSF satisfying a set of axioms arises as a result of the following steps: the performance of an agent is randomly mapped into a performance; the contest designer sets a performance threshold so as to satisfy a budget constraint; the agents whose performances exceed this cutoff win a prize. They refer to this class of CSFs as \textit{Random Performance Functions} (RPFs).

We contribute to this literature by considering the optimal contest design by a principal who is able to design the allocation rule. We show that RPFs are optimal under specific distributional goals of the contest designer. To the best of my knowledge, optimal contest design has been addressed only by \cite{moldovanusela} and \cite{letina_optimal_2023} in a setting with finitely many agents.

A recent literature systematically applies the study of extreme points of monotone functions to several economic settings; e.g., \cite{JP26}, \cite{kleiner_extreme_2021}, \cite{yang_monotone_2024} and \cite{yang_multidimensional_2025}. These papers consider settings in which the revelation principle applies, and hence the class of monotone functions studied maps the (possibly multidimensional) types of the agents into a payoff. In this paper, we consider allocation rules in contests, namely functions which are monotone in agents' effort, but which also depend on the overall distribution of effort exerted by all the other agents. Hence, the results from the papers cited above are not readily applicable, and a new result on extreme points of this class of monotone function is needed.

A similar model is considered in \cite{lagziel_screening_2022}. However, importantly in this paper we consider a large contest environment. Therefore, we can leverage a law-of-large-number type of argument to establish the optimality of threshold strategies, which are not necessarily optimal in \cite{lagziel_screening_2022}.

From a methodological perspective, a closely related paper is \cite{byeong25}. They consider a mechanism design setting and they derive the optimality of a two-tier quota mechanism, which is analogous to the one derived in this paper. However, the economic environment is fundamentally different, and it is not related to contests.

Lastly, there are contributions in Management exploring the benefits of randomizing selection in organizations.\footnote{See, \textit{inter alia}, \cite{Sto09}, \cite{Sto11}, \cite{FOR23}, \cite{OF20}, \cite{BOR20}. I thank Chengwei Liu for pointing me to this literature.} Interestingly, \cite{Liu21} derives the optimality of random selection, although in a very different framework.

\section{Extreme Points} \label{sec:extreme}
Let $\X \subset \mathbb{R}^{d}$ be a compact set and denote $\Delta(\X)$ the set of Borel measures on $\X$.

Since they will be the backbone of many later results, we recall the following important theorem.

\begin{theorem}[\cite{dubins_extreme_1962}; \cite{winkler_extreme_1988}]\label{winkler}
    For any compact and convex subset $A \subseteq \X$, for any family of affine functions from $A$ to $\mathbb{R}$, $\{\phi \}_{i=1}^{n}$,  and for any family of real numbers, $\{\gamma_i \}_{i=1}^{n}$, let $A_0=\{x \in A \ | \ \phi_i(x) \leq \gamma_i \}$. Then,
    \begin{equation*}
        \ext{A_0} \subseteq \left\{y \in A \ | \ y=\sum_{i=1}^{n+1} \lambda_ix_i, \ \sum_n \lambda_i=1, \ \lambda_i\geq 0, \ x_i \in \ext{A} \right\} \ .
    \end{equation*}
\end{theorem}
\medskip

Let $\F$ be the set of monotone functions
\begin{equation*}
    \F = \left\{ f: \X \times \Delta(\X) \longrightarrow[0,1] \ , \ f(\cdot, \mu) \ \text{increasing} \ \text{for every} \ \mu \in \Delta(\X)\right\} \ .
\end{equation*}
In this paper, we are interested in the extreme points of the set of functions
\begin{equation*}
    \F^*= \left\{ f \in \F \ :  \int_{\X} f(x, \mu) \ d\mu(x) = k  , \  k \in (0,1) \right\} \ .
\end{equation*}
By Theorem \ref{winkler}, the set of extreme points of $\F^*$ is a convex combination of at most two extreme points of $\F$; hence, we show that they are a mixture of at most two indicator functions defined on up-sets. Moreover, by a characterization of \cite{yang_multidimensional_2025}, these two up-sets are ordered by set inclusion. Therefore, we obtain the following characterization.
\begin{theorem}\label{thrm}
    The set of extreme points of $\F^*$ is given by
\begin{equation*}
    \ext(\F^*)= \left\{f^* \in \F^* : f^*(x, \mu)=\lambda\mathbb{I}_{A_1}+(1-\lambda)\mathbb{I}_{A_2} \right\}
\end{equation*}
for each $\mu \in \Delta(\X)$ and some $A_1 \subseteq A_2$. Moreover, $\F^*$ is compact with respect to the weak topology.
\end{theorem}
In other words, extreme points in $\F^*$ are step functions on $\X$ with at most two jumps.

\subsection*{Maximizing a Functional} Consider the case $\X=[0,1]$. We are interested in maximizing a functional $\Pi: \F^* \longrightarrow \mathbb{R}$; formally, we are interested in the following problem
\begin{equation} \label{eq:functional}
    \max_{f \in \F^*} \ \Pi(f) = \int_\X \Phi\left(f(x,\mu\right), x)  \ d\mu(x) \ .
\end{equation}
By \textit{Bauer's Maximum Principle}, we know that if $\Phi$ is quasiconvex in $f$, the maximizers of \eqref{eq:functional} can be found among the extreme points of the feasible set.
\begin{theorem}[Bauer's Maximum Principle \cite{ball2023bauersmaximumprinciplequasiconvex}] \label{bauer}
Let $A$ be a locally convex, topological vector space, $B$ be a non-empty compact, convex subset of $A$, and $S: B \longrightarrow \mathbb{R}$ be a upper semicontinuous, quasiconvex function. Then the set $E$ of maximizers of $S$ over $B$ is a face of $B$. Furthermore, $E$ contains an extreme point of $B$.
\end{theorem}
As in \cite{kleiner_extreme_2021}, define the majorization relation $\succ$ as follows.
\begin{definition}
    For each $\mu \in \Delta(\X)$ and for any two functions $f,g \in \F$, we say that $f$ \textit{majorizes} $g$, $f \succ g$ if
    \begin{equation*}
        \int_t^1 f(t, \mu) \ d\mu(t) \geq \int_t^1 g(t, \mu) \ d\mu(t) \ .
    \end{equation*}
\end{definition}
 Moreover, the set of extreme points $\F^*$ is naturally ordered with respect to the majorization order.

An elegant result due to \cite{FanLorentz} gives necessary and sufficient conditions for a functional to be monotonic in the majorization order.
\begin{theorem}[\cite{FanLorentz}] \label{thrmFanLorentz}
    Let $\Phi: \X \times \Delta(\X) \times \X \longrightarrow [0,1]$. Then,
    \begin{equation*}
        \int_X \Phi(f(x,\mu), x) \ d\mu(x) \geq \int_X \Phi(g(x,\mu), x) \ d\mu(x)
    \end{equation*}
    holds for any two functions $f,g \in \F$ such that $f \succ g$ if and only if $\Phi(u,t)$ is convex in $u$ and supermodular in $(u,t)$.
\end{theorem}

Theorems \ref{thrm}, \ref{bauer} and \ref{thrmFanLorentz} yield simple solutions to \eqref{eq:functional}. In the next sections, we show the usefulness of these results in the framework of large contests. 


\section{Large Contests with Complete Information} \label{sec:applications}
There is a measure-one mass of agents. They can choose an effort level $x\in \mathcal{X}=[\underline{x}, \overline{x}]$. Effort is perfectly observable. Define $\Delta(\X)$ as the set of probability measures on $\X$. A contest designer needs to choose an \textit{Allocation Rule}; i.e., a mapping $q: \X \times \Delta(\X) \longrightarrow [0,1]$. Following \cite{azrieli_success_2024}, we assume the designer has a budget sufficient to give prizes to a mass $k\in (0,1)$ of agents. Denote $\mu \in \Delta(\X)$ the potential distribution of efforts (the `competition') faced by an agent with effort $x$.

\begin{definition}
    For every $x\in \X$ and $\mu \in \Delta(\X)$, an \textit{Allocation Rule} $q:\X \times \Delta(\X) \longrightarrow [0,1]$ is non-decreasing in $x$ and satisfies
    \begin{equation} \label{eq:mktclrgen}
        \int q(x, \mu)d\mu(x)=k .
    \end{equation}
\end{definition}
The interpretation of (\ref{eq:mktclrgen}) is the following: For any distribution of efforts induced by a strategy profile, the total mass of winners will be equal to $k$. In other words, the prizes will be allocated, regardless of the equilibrium distribution of efforts.\footnote{All results go through even in the case where the designer can choose not to allocate a share of the prizes.}\\

Exerting effort has a cost $c(x)$, where $c(\cdot)$ is a convex function. Winning a prize  yields utility $V>0$. Agents' payoff function is given by
\begin{equation} \label{eq:agentcomplete}
 U(x)=q(x,\mu)V - c(x) \ .  
\end{equation}
In this setting, $q(x, \mu)$ is a \textit{CSF}\footnote{When effort is perfectly observable, the distinction between Allocation Rule and CSF is immaterial.}; i.e., for a given competition $\mu \in \Mprob$, it specifies the probability of winning a prize for an agent exerting effort $x$.

\subsection{Robust Contests} Denote $\mathbf{M}$ for a probability measure on $(\Mprob, \mathcal{B}(\Mprob))$.\footnote{Endowed with the topology of weak convergence, $\Mprob$ has a natural notion of $\sigma$-algebra. Formally, $\mu_n \rightarrow \mu \iff \int_{\X}g d\mu_n \rightarrow \int_\X g d\mu \ $ , for any $g$ continuous.}
The payoff of the designer is given by
\begin{equation} \label{payoffexp}
    \Pi=\int_{\Mprob} \int_\X q(x,\mu)\pi(x)\ d\mu(x) \ d\mathbf{M}(\mu) \ ,
\end{equation}
where $\pi(\cdot)$ is a continuous function.

Importantly, the specification in \eqref{payoffexp} assumes that the distribution of effort $\mu$ is \textit{exogenous}. This can be framed as robustness property in the design of the contest: What is the best allocation rule, \textit{regardless} of how agents will perform in the contest?\\
Moreover, \eqref{payoffexp} consistent with very different distributional preferences of the designer, depending on the the monotonicity of $\pi(x)$.

Suppose that $\pi(x)$ is increasing; intuitively, the designer cares relatively more about the better-performing agents. As an example, consider a private foundation sponsoring fellowships for students to pursue graduate studies. It might be argued that the foundation is interested in awarding these fellowships to the best-performing students\footnote{Of course, if performance is noisy, a best-performing student might not necessarily be one of the students who exerted the most effort. However, by a law-of-large-number type of argument (ensured by the weak convergence of measures), on average it will be the case that better-performing students will be those who have exerted more effort.}, as it has a stake in their future success, but not necessarily in all the other applicants'.

On the other hand, suppose that $\pi(x)$ is decreasing. This specification might reflect the designer's concerns about equal opportunity among competing agents. Alternatively, the designer might be concerned that agents have to exert extremely costly (or even wasteful) effort.\footnote{See \cite{OLSZEWSKI2019101} and the discussion therein about cram schools in Korea.}

The designer chooses an Allocation Rule satisfying \eqref{eq:mktclrgen} in order to maximize \eqref{payoffexp}. The following result readily follows from Theorems \ref{thrm}, \ref{bauer} and \ref{thrmFanLorentz}.
\begin{proposition} \label{optimalq}
    For any $\mathbf{M}$, the set of maximizers is such that 
    \begin{itemize}
        \item If $\pi(x)$ is increasing, the optimal Allocation Rule is given by
    \begin{equation} \label{optimalf}
q^*(\cdot, \mu) = \begin{cases}
    &1 \qquad x> \tilde{x} \\
    &0 \qquad x< \tilde{x}  \ ,
\end{cases}
\end{equation}
where, for each $\mu \in \Delta(\X)$, $\tilde{x} \in \X$ is the unique solution to \eqref{eq:mktclrgen};
\item If $\pi(x)$ is decreasing, the optimal Allocation Rule is given by $q^*(x, \mu)= \{\bar{q}\}$, where $\bar{q}(x,\mu)=k$, for any $x \in \X$ and $\mu \in \Delta(\X)$;
\item If $\pi(x)$ is non-monotonic, the optimal Allocation Rule is given by
\begin{equation} 
q^*(\cdot, \mu) = \begin{cases}
    &1 \qquad x> x'' \\
    &\alpha \qquad x' \leq x \leq x''\\
    &0 \qquad x< x'  \ ,
\end{cases}
\end{equation}
for some $x',x'' \in\X$ and $\alpha \in [0,1]$.
    \end{itemize}
\end{proposition} 

Proposition \ref{optimalq} characterizes the optimal allocation rules. It provides an economic foundation for the allocation rules discussed in the Introduction.

\subsection{Endogenous Effort Distribution} \label{sec:endogenous}
Consider now the case of the distribution of effort being endogenous. In other words, agents strategically optimize their effort level in response to the allocation rule proposed by the designer. The objective function is now given by 
\begin{equation}\label{genpayoff}
    \Pi(q,\mu)=\int_{\X} \Phi(q(x,\mu),x) \ d\mu(x) \ ,
\end{equation}
where $\Phi(\cdot)$ is a continuous function. However, in this section $x^*$ maximizes \eqref{eq:agentcomplete}. Denote $\max_q \Pi(q,\mu) \ \text{s.t.} \eqref{eq:agentcomplete}$ as the \textit{primal problem}.

From a methodological perspective, the effort distribution $\mu$ now depends on the allocation rule $q$. Therefore, the payoff of the designer is not linear in $q$, and the analysis from the previous subsection cannot be extended to the current setting. However, the main contribution of this subsection is to develop a methodology to reformulate the designer's problem so as to apply Theorem \ref{winkler}, and derive very similar results to the case of robust contests.

In order to do this, we need to redefine the feasible set in terms of rents. Denote $W \in \mathbb{R}$ as the rent the designer is leaving to the agent. For a given $W$, indifference between exerting effort $x$ at cost $c(x)$, with a probability $g(x)$ of winning the prize $v$ has to satisfy $g(x)V-c(x)=W$. Define

\begin{equation}
    g_W:=\frac{c(x)+W}{V}, \qquad \X_W:=\{x \in \X \ : \ c(x) \leq V - W \} \ .
\end{equation}
In other words, $g_W$ is the winning probability which induces the agent to exert effort $x$ when leaving a rent $W$. Since it is a probability, it must be the case that $g_W \leq 1$, or $c(x) \leq V-W$. Hence, $X_W$ is the set of feasible effort levels at a given rent $W$. Since agents are homogeneous, in equilibrium every effort chosen must yield the same rent $W$. Therefore, on the support of any candidate equilibrium distribution of effort $\mu^*$, the allocation rule \textit{must} be equal to $g_W$. Moreover, it holds that $\mu^*(X_W)=1$.

Thus, we can redefine the set of allocation rules as
\begin{equation*}
    \F^W = \left\{ (W, \mu) \in [0, V] \times \Delta(\X): \ \mu(X_W)=1, \ \int_{\X} g_W \ d\mu=k \right\} \ ,
\end{equation*}
and the objective function as 
\begin{equation*}
    J(W, \mu) := \int_{\X} \Phi(g_W(x),x) \ d\mu(x) \ .
\end{equation*}
We refer to $\max_W J(W, \mu) \ \text{s.t.} \ \mathcal{F}_W$ as the \textit{dual problem}.
Lastly, we need the following definition.

\begin{definition} \label{def:endogenous}
    For any non-decreasing function $g: \X \longrightarrow [0,1]$, denote
\begin{equation*}
    U_g(x):=g(x)V - c(x), \qquad A(g):=\left\{x \in \X \ : \ U_g(x)=\sup_{\X} U_g\right\} \ .
\end{equation*}
We say $\mu \in \Mprob$ is an equilibrium of $q \in \mathcal{Q}$ if $\mu(A(\cdot, \mu))=1$.
\end{definition}
Define the designer's value as
\begin{equation*}
    \Pi^*:=\sup \left\{ \Pi(q,\mu) \ : \ q \in \mathcal{Q}, \ \mu \ \text{is an equilibrium of} \ q, \ \int_\X q(x, \mu) d\mu=k, \ k \in (0,1) \right\} \ .
\end{equation*}
The following theorem shows that it is without loss of generality to consider the dual problem. The advantage is that, for any given $W$, the functional is linear in $\mu$ and the feasible set $F_W$ is convex. Therefore, Theorem \ref{winkler} can be applied in the space of probability measures in order to derive the equilibrium distribution of efforts preferred by the designer, and hence, the optimal allocation rules which induce them.

\begin{theorem}\label{thm:endogenous}
For any continuous function $\Phi$, it holds that $\Pi^*=\sup_{\F^W} J(W,x)$, and the supremum is attained. Moreover, there exists an optimal pair $(q^*, \mu^*)$, a value $W^* \in [0, kV]$, and points $y_1 \leq y_2$ in $X_{W^*}$ such that
\begin{enumerate}
    \item[(i)] $\mu^*=\lambda \delta_{y_1} + (1-\lambda)\delta_{y_2}$ for some $\lambda \in [0,1]$,
    \item[(ii)] The optimal allocation rule satisfies \begin{equation*}
        q^*(x, \mu^*)= \begin{cases}
        0 \qquad \qquad \ \  \text{for} \  x<y_1 \\
        g_{W^*}(y_1) \qquad \text{for} \ y_1 \leq x \leq y_2 \\
         g_{W^*}(y_2) \qquad \text{for} \ x \geq y_2 \\
    \end{cases}
    \end{equation*}
    \item[(iii)] Every agent's equilibrium payoff is equal to $W^*$ \ .
\end{enumerate}
    
\end{theorem}

Theorem \ref{thm:endogenous} characterizes the set of optimal allocation rules with endogenous effort distribution. Strikingly, it displays the same structure as in the robust case; the optimal allocation rule is a step function with at most two steps. Agents will concentrate their effort levels at the threshold values $y_1$ and $y_2$. Since agents are homogeneous, in equilibrium they earn the same rents $W^*$.

The points $y_1$ and $y_2$ are pinned down by the functional form of $\Phi(\cdot)$.

\begin{theorem} \label{thrm:endogenous2}
    Let $\Phi(q,x)=q \pi(x)$, with $\pi(\cdot)$ continuous. Denote $x_m \in \argmax_{x} \pi(x)$ and let $\alpha:= \max\{k, c(x_m)/V\} \in [k,1]$. Then,
    \begin{enumerate}
        \item[(i)] $\Pi^*=k\pi(x_m)$,
        \item[(ii)] $\Pi^*$ is attained by every $q \in \mathcal{Q}$
        \begin{align*}
            q^*(\cdot, \mu^*)&=\begin{cases}
                \alpha \qquad x \geq x_m \\
                0 \qquad x< x_m \ ,
            \end{cases}\\
            \mu^*&= \frac{k}{\alpha}\delta_{x_m} + (1-\frac{k}{\alpha})\delta_{\underline{x}} \ .
        \end{align*}
        
    \end{enumerate}
\end{theorem}
Intuitively, when effort is perfectly observable and the designer cares only about the winners, she can set the threshold values so that agents will exert precisely the payoff-maximizing effort levels.

When the designer cares relatively more about high-performing agents; i.e., $\pi(\cdot)$ is increasing, she sets a threshold at the highest possible effort level. As the next results shows, agents are then indifferent between exerting $\Bar{x}$ and staying out of the contest.

\bigskip

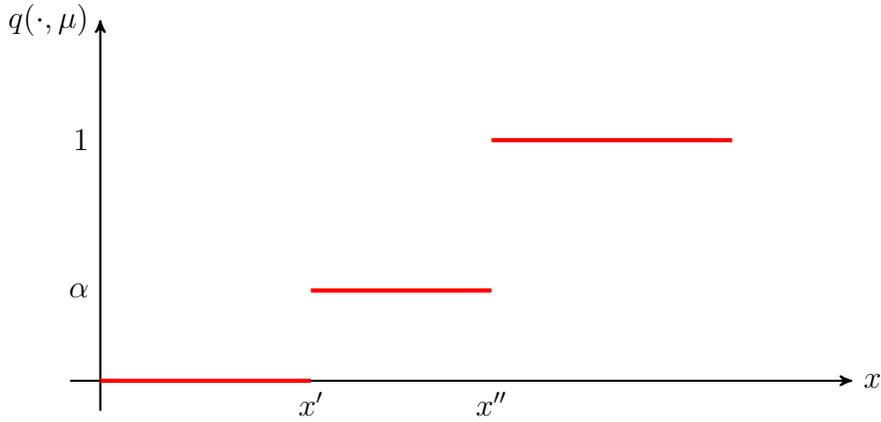
\begin{figure}[ht]
\centering
    \begin{tikzpicture}[-,>=stealth',auto,node distance=1.2cm,
  thick,main node/.style={transform shape, circle,draw,font=\rmfamily, minimum width=5pt}, scale=0.8]

    \draw[->] (-5,0) -- (8,0) node[right]{$x$};
    \draw[->] (-4.5,-0.5) -- (-4.5,6) node[left]{$q(\cdot, \mu)$};

  \draw[ultra thick, red] (-4.5,0) -- (-1,0);
  \draw[ultra thick, red] (-1,1.5) -- (2,1.5);
  \draw[ultra thick, red] (2,4) -- (6,4);
   \draw (-4.5,1.5)node[left]{$\alpha$};
   \draw (-1,0)node[below]{$x'$};
   \draw (2,0)node[below]{$x''$};
   \draw (-4.5,4)node[left]{1};
\end{tikzpicture}
\caption{An optimal allocation rule with two steps.}
\label{fig:allocationrule}
\end{figure}

\begin{proposition}
    Suppose $\pi(x)$ is increasing. Then, the (essentially) unique equilibrium is for a share $k$ of agent to exert maximal effort $x^*=\overline{x}$; and for the remaining agents to stay out of the contest.
\end{proposition}
\begin{proof}
Define $\Bar{x}$ as $V=c(\Bar{x})$. Suppose there is an interior equilibrium effort level $\tilde{x}<\overline{x}$. Then, from \eqref{eq:mktclrgen} we know that
\begin{align*}
    \int_\mathcal{X}q^*(x,\delta_{\tilde{x}})\ d\delta_{\tilde{x}}(x) &= \\
    q^*(\tilde{x},\delta_{\tilde{x}}) &= k \ .
\end{align*}
Hence, the optimal $q^*$ is given by
\begin{equation} \label{optimalfdirac}
q^*(\cdot, \mu) = \begin{cases}
    &1 \qquad x > \tilde{x} \\
    &k \qquad x =\tilde{x}  \\
    &0 \qquad x < \tilde{x}  \ ,
\end{cases}
\end{equation}
But then, increasing the effort to $\tilde{x}+ \epsilon$ yields a discontinuous jump in payoff, and therefore $\tilde{x}$ cannot constitute an equilibrium effort level. Moreover, all agents exerting maximal effort cannot constitute an equilibrium either. In fact,
\begin{align*}
        U(x) &= q^*(\Bar{x}, \delta_{\Bar{x}})V-c(\Bar{x}) \\
             &=kV-c(\Bar{x})\\
             &<V-c(\Bar{x})\\
             &=0
    \end{align*}
Then, the only equilibrium can thus be $\overline{x}$ for a share $k$ of agents, whereas the remaining $1 - k$ stay out of the contest. All agents earn zero payoff.
\end{proof}
Intuitively, the discontinuity at the threshold value $\tilde{x}$ in \eqref{optimalfdirac} gives agents an incentive to slightly increase effort so as to gain a discontinuous jump in the probability of winning. Therefore, the contest designer induces the agents to increase effort until their utility gets to zero: either they stay out, or they exert maximal effort.

Similarly, when $\pi(x)$ is non-monotone, the designer might prefer agents to exert some intermediate level of effort. Then, by setting the two thresholds appropriately, the designer can induce agents to do so by extending the jump in the middle, so as not to give them the incentive to slightly increase their effort level, for any given competition $\mu$ they face.

\subsection{More General Objective Function} \label{sec:extensions}
A common specification for the designer's payoff is given by
\begin{equation} \label{otherpayoff}
    \Pi'= \int_\X \pi'(x)\ d\mu(x)  \ ,
\end{equation}
where $\pi'(x)$ is a continuous function.\\
In other words, the designer might be interested in maximizing some function of aggregate effort, and s/he might also be concerned with the agents who do not win a prize.

Let $\psi=\pi' \circ c^{-1}$ defined on $[0, c(\Bar{x})]$. We have the following result.

\begin{proposition}\label{prop:continuous}
Consider the payoff of the designer as in $\eqref{otherpayoff}$.
    \begin{itemize}
        \item[(i)]  If $\psi$ is concave, then
        \begin{align*}
            q^*(\cdot, \mu^*)&=\begin{cases}
                k \qquad x \geq y^* \\
                0 \qquad x< y^* \ ,
            \end{cases}\\
            \mu^*&= \delta_{y^*}  \ ,
        \end{align*}
        where $y^*=c^{-1}(kV)$;
        \item[(ii)] If $\psi$ is convex, then
        \begin{align*}
            q^*(\cdot, \mu^*)&=\begin{cases}
                1 \qquad x = \Bar{x} \\
                0 \qquad x< \Bar{x} \ ,
            \end{cases}\\
            \mu^*&= k\delta_{\overline{x}} + (1-k)\delta_{\underline{x}} \ .
        \end{align*}
        In both cases, agents earn no rents: $W^*=0$.
    \end{itemize}
\end{proposition}

A common specification is that the designer wants to maximize aggregate effort: $\pi'(x)=x$. Proposition \ref{prop:continuous} shows that the best the designer can do is setting a threshold $y^*$ and award prizes \textit{at random} for performances above that threshold. Interestingly, this result does not hinge on the imperfect observability of agents' effort.

In other words, even if effort was perfectly observable, the designer finds it optimal to introduce a lottery in order to maximize aggregate effort, and not just the effort of the winners.
\subsection{Heterogeneous Agents}
Consider now the case in which agents are heterogeneous with respect to their cost of effort. Each agent has cost of effort $c_i$, for $i \in \{1, \dots, n\}$. Each type $i$ has mass $q_i>0$, with $\sum_i q_i=1$. An aggregate distribution of effort is then given by $\mu=\sum_iq_i \mu_i$, where $\mu_i$ is the distribution of effort for agents with type $i$. An allocation rule $q \in \mathcal{Q}$ is a function of the aggregate effort distribution $\mu$.\footnote{If the allocation rule is allowed to be type-specific, then the problem separates across each component $i$, as in Section \ref{sec:diversity}.} We can then define $U_{g,i}(x)$ and $A_i(g)$ as in Definition \ref{def:endogenous} after indexing the cost by types.

Lastly, define
\begin{equation*}
    g_i^{W_i}:=\frac{c_i(x)+W_i}{V}, \qquad \X_i(W_i)=\{ x \in \X \ : \ c_i(x) \leq V-W_i, \ c_i(x) + W_i \leq c_j(x)+W_j \ , \  \forall j \neq i  \} \ .
\end{equation*}
In other words, $X_i(W_i)$ is the set of effort levels which are feasible only for agents of type $i$. Analogously, we define
\begin{equation*}
    J_n(W, (\mu_i)):= \sum_i q_i \int_X \Phi_i(g_i^{W_i}(x),x) \ d\mu_i \ ,
\end{equation*}
and
\begin{equation*}
     \F_n = \left\{ (W, \mu) \in [0, V]^n \times \Delta(\X): \ \mu(E_i(W))=1 \ \forall i, \ \sum_i q_i\int_{\X} g_i^{W_i} \ d\mu=k \right\} \ .
\end{equation*}
Along the lines of Theorem \ref{thm:endogenous}, we have the following result.
\begin{theorem}
    For any continuous $\Phi_n$, the supremum $\Pi_n^*$ is attained and it equals the value of the dual $\Pi_n^*=\max_{\F_n} J_n$. There exists an optimal pair $(W^*, (\mu_i^*))$ such that the aggregate distribution $\mu^*$ has at most $n+1$ atoms, and the optimal allocation rule $q^*(x, \mu^*)$ is a step function with at most $n+1$ steps.
\end{theorem}
Qualitatively, the results are analogous to the ones derived in Section \ref{sec:endogenous}. However, the optimal allocation rule can now have up to $n+1$ steps.

\section{Large Contests with Incomplete Information}
In this section, we consider contests with incomplete information. Effort is not perfectly observable; performance $y$ is a function of effort $x$ and a noise term $\epsilon$. For the design of robust contests, the characterization is identical to the one of robust contests with complete information, and no further assumption is needed. On the other hand, when the distribution of effort is endogenous, stronger assumptions are needed to derive the optimal allocation rules.
\subsection{Robust Contests}
Fix a given distribution of effort $\mu$. An allocation rule is a non-decreasing function $g:\mathbb{R} \longrightarrow [0,1]$. Each allocation rule induces a \textit{Contest Success Function (CSF)} $F_g(x, \mu)=\int g(y) f(y |x) d \lambda (y)$. Each effort level $x$ induces a performance distribution $\nu_\mu(y)=\int f(y | x) d \mu(x) d \lambda(y)$. The budget is now given by $\int F_g d\mu= \int g d \nu_\mu=k$, where the first equality follows by Fubini's theorem. The objective function is given by
\begin{equation*} \label{obj:incompleterob}
    \Phi= \int F_g(x, \mu) \pi(x) d \mu(x) = \int g(y) \hat{\pi}_\mu(y) d \nu_\mu(y) \ ,
\end{equation*}
where
\begin{equation*}
    \hat{\pi}_\mu(y)= \frac{\int \pi(x)f(y|x) d\mu(x)}{\int f(y|x) d\mu(x)} \ .
\end{equation*}
If $\pi(x)$ is monotonic, $\hat{\pi}(x)$ is not necessarily so. The following assumption is standard, and it ensures that this is the case.
\begin{assumption}[MLRP]\label{ass:MLRP}
The family of distributions $F(\cdot |x)$ admits densities $f(\cdot|x)$ with respect to a common measure $\lambda$. Moreover, for all $x'>x$, the ratio $f(y|x')/f(y|x)$ is non-decreasing in $y$ whenever the denominator is positive.
\end{assumption}
Upon changing the variable from $q(x, \mu)$ to $g(y)$, choosing the optimal $g(y)$ subject to the budget constraints is the same as in the complete information case. Hence, we have the following result.
\begin{proposition}
    The optimal allocation rule $g^*(y)$ is a step function with at most two steps. If Assumption \ref{ass:MLRP} holds, the characterization is equivalent to the one of Proposition \ref{optimalq}.
\end{proposition}

\subsection{Endogenous Effort Distribution}
When agents best-respond to the allocation rule and their performance is noisy, the designer's problem loses tractability in the general case. There are three main differences with respect to the complete information case.

First, for a given rent $W$, the indifference condition $g(x)V-c(x)=W$ is no longer sufficient to pin down the equilibrium allocation rule. This is because, when performance is noisy, for each equilibrium distribution of effort, the allocation rule has to satisfy a first-order condition for each agent. The two conditions determine the level and the slope of the allocation rule, and they are independent functionals.

Second, under complete information the designer can lower the allocation rule off the equilibrium path at any effort $x$ without modifying it anywhere else. In other words, as long as the allocation rule is non-decreasing, every `deviation' can be dealt with independently. When performance is noisy, the CSF $F_g(x,\mu)$ is an average rule, weighted by a density with full support. Therefore, lowering $g$ on any set of positive measure lowers the CSF \textit{at every effort}. 
Therefore, `blocking' a deviation at effort $x$ changes the level and slope of $g$; the indifference (or global incentive compatibility) condition $vF_g(x,\mu) \leq c + W$ becomes a continuum of constraints which must be satisfied simultaneously to the local constraints.

Third, the feasible set at a given rent $W$, $\mathcal{F}_W$, is no longer convex; the pointwise maximum of two induced CSFs is not itselfs induced by any allocation rule. Hence, mixtures of equilibria of different allocation rules do not necessarily constitute equilibria.

Therefore, we adopt the theoretical framework of \cite{morgan_limits_2022}, which preserves tractability and allows to get interesting insights into the optimal design of allocation rules.

\paragraph{\cite{morgan_limits_2022}} Agents choose an effort $e^x$, with $x \in [-\infty, \infty)$. Each effort level $x$ costs $C(e^x)=: c(x)$. We normalize $c(-\infty)=0$. We make the following two assumptions.

\begin{assumption}\label{ass:cost}
    $C:[0, \infty) \longrightarrow [0, \infty)$ is continuous on $(0, \infty)$, twice continuously differentiable and it satisfies $C(0)=0$, $C'>0$, and $C''>0$. The elasticity of the marginal cost $\eta_C(q)=qC''(q)/C'(q)$ is non-decreasing and $qC'(q) \longrightarrow \infty$ as $q \longrightarrow \infty$.
\end{assumption}

The performance of an agent exerting effort $x$ is given by $y=x+ \epsilon$, where $\epsilon$ is drawn independently from a distribution $F$.

\begin{assumption}\label{ass:density}
    $F$ admits a density $f$ which is continuously differentiable, strictly increasing and strictly log-concave on $\mathbb{R}$.
\end{assumption}

Assumption \ref{ass:density} implies that the hazard rate $h=f/1-F$ is strictly increasing. This is an important assumption for deriving the results below. Noise drawn from a normal or logistic distribution satisfies Assumption \ref{ass:density}.

\subsubsection*{Allocation Rules} An allocation rule on performance is a non-decreasing function $g: \mathbb{R} \longrightarrow [0,1]$. Denote the set of such functions as $\mathcal{G}$. Denote $g_0:= \lim_{y \longrightarrow \infty} g(y)$ and $\nu$ for the Lebesgue-Stieltjes measure of $g$. An agent exerting effort $x$ wins with probability

\begin{equation*}
    P_g(x):= \int g(x+e)f(e) \ de = g_0 + \int \Bar{F}(t-x) \nu(t) \ , \qquad P'_g(x)=\int f(t-x) \nu(t) \ ,
\end{equation*}
with $P_g(- \infty)=0$. By dominated convergence, $P_g$ is non-decreasing and continuously differentiable on $\mathbb{R}$. Intuitively, $g$ is the distribution function of a random cutoff $T$ which is independent of $\epsilon$; $P_g(x)=Pr(x+ \epsilon \geq T)$.

\subsubsection*{Equilibrium} For any distribution of effort $\mu \in \Delta(X)$, an allocation rule is a function $g_\mu \in \mathcal{G}$ such that $\int P_{g_{\mu}} \ d\mu=k$. An agent exerting effort $x$ has utility $U_g(x):=V P_g(x)-c(x)$, with $U_g(-\infty)=0$. We have the following definition.

\begin{definition}
    Denote
    \begin{equation*}
        A(g):= \{ x: U_g(x)= \sup_X U_g(x)\} \ .
    \end{equation*}
    The pair $(g, \mu)$ constitutes an equilibrium if
    \begin{equation*}
        \mu(A(g))=1 \qquad \text{and} \qquad \int P_g \ d\mu=k \ .
    \end{equation*}
\end{definition}

The designer wants to maximize aggregate effort:
\begin{equation*}
    \max_{(g,\mu)} O(\mu):=\int e^x \ d\mu(x) \ ,
\end{equation*}
subject to the pair $(g, \mu)$ being an equilibrium.

Before stating the results, we need a last piece of notation. For a given probability $\phi \in (0,1)$, denote $\kappa(\phi):=f(\Bar{F}^{-1}(\phi))$ as the density of the noise at the quantile at which an agent wins the prize with probability $\phi$. Accordingly, $x(\phi):=(c')^{-1}(V\kappa(\phi))$ is the optimal effort when the winning probability is $\phi$. Lastly, define
\begin{equation*}
    \Lambda(\phi):= V\phi-c(x(\phi))
\end{equation*}
as the utility of the agent when exerting effort $x(\phi)$. We are ready to state the main result of this Section.

\begin{theorem}\label{thrm:MTV}
    Let Assumptions \ref{ass:cost} and \ref{ass:density} hold. Then,
    \begin{itemize}
        \item[(i)] For every equilibrium pair $(g, \mu)$,

        \begin{equation*}
            O(\mu) \leq \min\{C^{-1}(kv), e^{x(k)} \} \ ;
        \end{equation*}
        \item[(ii)] If $\Lambda(k) \geq 0$, then $C^{-1}(kv) \geq e^{x(k)}$, and there exists a threshold value $t^* \in \mathbb{R}$ such that the unique equilibrium pair is given by
         \begin{align*}
            g^*&=\begin{cases}
                1 \qquad y \geq t^* \\
                0 \qquad y< t^* \ ,
            \end{cases}\\
            \mu^*&= \delta_{x(k)}  \ ;
        \end{align*}
        \item[(iii)] If $\Lambda(k) < 0$, then $C^{-1}(kv) < e^{x(k)}$, and there exists $\alpha \in (0,1)$ and a threshold value $t^o \in \mathbb{R}$ such that the unique equilibrium pair is given by
         \begin{align*}
            g^*&=\begin{cases}
                \alpha \qquad y \geq t^o \\
                0 \qquad y< t^o \ ,
            \end{cases}\\
            \mu^*&= \delta_{x^o}  \ ,
        \end{align*}
        where $x^o$ denotes the first-best effort level.
    \end{itemize}
\end{theorem}
Theorem \ref{thrm:MTV} has a simple interpretation in terms of the first-best effort level $x^o$. Suppose that every agent exerts $x^o$ and the designer awards the prizes to the best performers. The budget constraint implies that there will be some threshold $t$ such that $\Bar{F}(t-x^o)=k$. An agent's marginal benefit from effort $x^o$ is given by the density of the noise at the threshold, namely $Vf(\Bar{F}^{-1}(k)) \equiv v \kappa(k)$; the marginal cost equals $c'(x^o)$. There are two possibilities.

If $V\kappa(k) \leq c'(x^o)$, the designer would like to increase the slope of the marginal incentive to exert effort. However, she cannot: conditional on having a share $k$ of winners, the winning probability for the best performers cannot exceed one. Therefore, in equilibrium agents exert effort $x(k)<x^o$, equating marginal benefit and marginal cost. All agents earn rents $\Lambda(k)>0$.

If $V\kappa(k)> c'(x^o)$, the effort which equates marginal benefit and marginal cost is higher than $x^o$. If the designer were to award prizes to the best performers, each agent exerting effort $x^o$ would have an incentive to increase effort in order to win with a higher probability. The market clearing threshold increases until agents are indifferent between participating with high effort level or staying out of the contest. Since costs are convex, the loss in aggregate effort due to drop-outs is higher than the extra effort exerted by the agents who participate. Unlike in the previous case, the designer can now decrease the slope of the marginal benefit of effort by introducing randomization; by scaling downs the probability of winning by a factor $\alpha<1$, the designer can restore equality between marginal benefit and marginal cost of effort at $x^o$ with full participation.

Crucially, the case distinction arises by comparing $h(\Bar{F}^{-1}(k))$, the hazard rate of the noise at the winning quantile, with the elasticity cost at the first-best effort level, $c'/c$ at $x^o$. Importantly, a more precise measurement of effort raises the former and induces the designer to \textit{prefer} the lottery. In fact, consider the following assumption from \cite{morgan_limits_2022}.

\begin{assumption}\label{ass:sigma}
    Let $\epsilon= \sigma \Tilde{\epsilon}$, and assume $\Tilde{\epsilon}$ satisfies Assumptions \ref{ass:density}, with $\sigma>0$.
\end{assumption}
\cite{morgan_limits_2022} refer to $\sigma$ as the `meritocracy', or the precision of the measurement of performances. We have the following corollary of Theorem \ref{thrm:MTV}.

\begin{corollary}\label{cor:sigma}
Suppose Assumptions~\ref{ass:cost}, \ref{ass:density} and~\ref{ass:sigma}
hold, with $C''>0$, and let $\kappa$ be computed from the standardized
noise $\tilde\varepsilon$. Define
\[
\underline\sigma:=\frac{v\,\kappa(k)}{c'(x^{\circ})}
=\frac{v\,f\big(\bar F^{-1}(k)\big)}{c'\big(c^{-1}(kv)\big)} .
\]
Then the deterministic threshold of Theorem~\ref{thrm:MTV}(ii) is optimal
if and only if $\sigma\ge\underline\sigma$, and the lottery of
Theorem~\ref{thrm:MTV}(iii) is optimal if and only if
$\sigma<\underline\sigma$. The optimal output is
\[
O^*(\sigma)=\min\big\{C^{-1}(kv),\,e^{x(k;\sigma)}\big\},\qquad
x(k;\sigma)=(c')^{-1}\big(v\kappa(k)/\sigma\big),
\]
which is non-increasing in $\sigma$ and equals the first best
$C^{-1}(kv)$ exactly for $\sigma\le\underline\sigma$.
\end{corollary}

This goes against the rationale adopted by the funding awards from the British Academy: a lottery is not a device for awarding prizes to candidates whose quality of performances cannot be ranked with respect to each other. On the contrary, the more precisely candidates can be assessed, the stronger the rationale for a lottery, as it raises effort even from those who in the end will not win the contest.

On the other hand, when performance is sufficiently noisy, awarding prizes at random will result in low aggregate effort. Therefore, aggregate effort is maximized by raising the expected benefits of high performers, even though this will result in a share of agents staying out of the contest.

Theorem \ref{thrm:MTV} and Corollary \ref{cor:sigma} provide a theoretical foundation for the \textit{Random Performance Functions} axiomatized in \cite{azrieli_success_2024}, provided the noise is sufficiently large: The optimal allocation rule induces a CSF which takes the form of a RPF.\footnote{This result extends to any increasing function $\pi'(x)$, and not just to the case in which the designer wants to maximize aggregate effort.}

\section{Extensions to Robust Contests} \label{sec:extensions}

\subsection{Multiple Prize Levels} Suppose there are different prize levels $\mathbf{k}=\{k_1, k_2, \dots k_n \}$, such that $\sum_i^nk_i=k \in (0,1)$, with respective values $V_1, V_2, \dots, V_n$. The following result shows that it is without loss of generality for the contest designer to focus on a single prize.

    \begin{lemma} \label{lemma:concavity}
    $k \mapsto \Pi(k)$ is concave.
\end{lemma}
\begin{proof}
Denote $\F^{o}(k)$ as the set of maximizers. Since $\F^*$ is convex, we have that, for $f_1 \in \F^{o}(k_1)$ and $f_2 \in \F^{o}(k_2)$, and any $k_1, k_2 \in (0,1)$ 
\begin{align*}
    \Pi(\lambda k_1 + (1-\lambda) k_2) &\geq \Pi(\lambda f_1 + (1-\lambda)f_2) \\
    &=\lambda \Pi(f_1) + (1-\lambda)\Pi(f_2) \\
    &= \lambda \Pi(k_1) + (1-\lambda) \Pi(k_2) \ .
\end{align*}
\end{proof}

\subsection{Large Contests with Diversity Constraints} \label{sec:diversity}
Consider now the following variation of the model. In addition to choosing an effort level $x \in \X$, suppose agents have an observable characteristic $\theta \in \Theta=\{\theta_1, \theta_2\}$, with relative shares $\alpha_1, \alpha_2$, respectively.
Assume the contest designer faces an additional constraint to give a share $h$ of the prizes to the subset of agents with characteristic $\alpha_1$:
\begin{equation*}
    \int_{\X} q(x, \mu_{|\theta_1}) \ d \mu_{|\theta_1}(x)= h \ ,
\end{equation*}
with $k > \alpha_1 > h$.

This additional constraint might be interpreted as a diversity constraint, or a quota on the number of agents with a given characteristic.

For example, at EPFL in Switzerland, the number of students admitted to the Bachelor's programmes has recently been capped at 3000: Swiss students holding the relevant qualifications will be admitted without limitations, whereas foreign students will be admitted according to the number of remaining places.\footnote{See \textcolor{blue}{https://actu.epfl.ch/news/admissions-to-bachelor-programs-will-be-limited-st/}.}

The payoff of the designer can now be written as
\begin{align*} 
    \Pi&=\int_{\Mprob} \int_\X q(x,\mu)\pi(x)\ d\mu(x) \ d\mathbf{M}(\mu)\\
    &= \int_{\Mprob} \alpha_1 \int_\X q(x,\mu_{|\theta_1})\pi(x)\ d\mu_{|\theta_1}(x)+ \alpha_2 \int_\X q(x,\mu_{|\theta_2})\pi(x)\ d\mu_{|\theta_2}(x) \ d\mathbf{M}(\mu)\\
    &= \int_{\Mprob} \alpha_1 \Pi(\mu_{|\theta_1}) + \alpha_2 \Pi(\mu_{|\theta_2}) \ d\mathbf{M}(\mu) \ .
\end{align*}
The designer's problem is then equivalent to two subproblems, one for each subset of agents $\mu|_{\theta_i}$, $i=1,2$.

\section{Conclusion} \label{sec:conclusions}
This paper characterizes the set of extreme points of a class of multidimensional monotone functions and applies it to large contests. Several directions for future research remain to be pursued. First, from a mathematical perspective the results strongly follow from the additive separable specification of the agents' utility function in \eqref{eq:agentcomplete}. Assuming a more general functional form might significantly complicate the analysis. Understanding under which conditions the results are robust to a more general specification is an interesting endeavour, and it is left for future research.

Second, from a more applied perspective, the dramatic surge in submission proposals in research funding contests has been associated to the extensive use of AI, which might bring the cost of effort close to zero. Analyzing the optimal design of contests in which effort and AI are complements or substitutes in determining agents' performance requires a different model, and it also left for future research.
\pagebreak
\appendix
\section{Proofs}
\subsection{Proof of Theorem \ref{thrm}}
\begin{proof} Define $\mathcal{F'}=\left\{ f:\X \longrightarrow [0,1], f   \ \text{increasing} \right\}$.
\begin{step}
    The set of extreme points of $\F$.
\end{step}
By \cite{Choquet1954}, the set of extreme points of $\F'$ is given by the functions $f=\mathbb{I}_{[a,1]}$ (or $g=\mathbb{I}_{(a,1]}$), where $\mathbb{I}_A$ denotes  the indicator function on $A$. The following lemma relates the extreme points of $\F'$ and $\F$.
\begin{lemma}
    Suppose $f$ is extremal in $\mathcal{F}$. Then, $f(\cdot, \mu$) is extremal in $\mathcal{F'}$.
\end{lemma}
\begin{proof}\footnote{I thank Stefan Wenger for suggesting this beautiful proof.}
Let $f$ be extremal in $\mathcal{F}$. Suppose that $\exists \mu_0$ such that $f(\cdot, \mu_0)$ is not extremal in $\mathcal{F'}$. Then, there exist $g_0$, $g_1 \in \mathcal{F'}$ and $t \in (0,1)$ such that $f(\cdot, \mu_0)=tg_0 + (1-t)g_1$. Define
\begin{equation}
f_i(\cdot, \mu)= \begin{cases}
    &f(\cdot, \mu) \qquad \mu \neq \mu_0 \\
    &g_i(\cdot) \qquad \ \ \mu=\mu_0
\end{cases}
\end{equation}
for $i=0,1$. But then we have that $f_0 \neq f_1$ and $f=tf_0 + (1-t)f_1$. Hence, $f$ is not extremal.
\end{proof}
By Theorem \ref{winkler}, an extreme point of $\F^*$ is a convex combination of at most two extreme points of $\F$. Hence, extreme points of $\F^*$ are of the form
\begin{equation*}
    f^*=\lambda\mathbb{I}_{A_1}+(1-\lambda)\mathbb{I}_{A_2} \ .
\end{equation*}
Finally, Proposition 3.1 in \cite{yang_multidimensional_2025} gives a characterization of monotone functions as probability distributions over `nested' up-sets; i.e., up-sets ordered by set inclusion. Therefore, $A_1 \subseteq A_2$.
\begin{step}
    $\F$ is a subset of a locally convex space.
\end{step}
For $\mu\in\Delta(\X)$, let $\F'_\mu\subset L^1(\X,\mu)$ denote the set of
$\mu$-equivalence classes of increasing functions from $\X$ to $[0,1]$. Let
\begin{equation*}
    V=\prod_{\mu\in\Delta(\X)}L^1(\X,\mu)
\end{equation*}
be endowed with the product of the norm topologies. Then $V$ is a locally
convex Hausdorff topological vector space. We denote by $V^*$ its dual and by
$\sigma(V,V^*)$ its weak topology; $\bigl(V,\sigma(V,V^*)\bigr)$ is again a
locally convex Hausdorff space.

We identify $f\in\F$ with the family
$\bigl(f(\cdot,\mu)\bigr)_{\mu\in\Delta(\X)}$, where two functions are
equivalent if they are agree $\mu$-a.e. for every
$\mu\in\Delta(\X)$. With this identification, we have that
\begin{equation*}
    \F=\prod_{\mu\in\Delta(\X)}\F'_\mu\subset V .
\end{equation*}

\begin{step}
   $\F$ is convex and compact in $\bigl(V,\sigma(V,V^*)\bigr)$.
\end{step}
Convexity holds because a convex combination of increasing $[0,1]$-valued
functions is again increasing and $[0,1]$-valued. Compactness is established in 
Lemma~\ref{lem:F-compact} below. It relies on the following identification of
the weak topology of the product space.

\begin{lemma}\label{lem:weak-product-F}
On $\F=\prod_{\mu\in\Delta(\X)}\F'_\mu\subset V$, the topology induced by the
weak topology $\sigma(V,V^*)$ coincides with the product of the topologies
induced on the sets $\F'_\mu$ by the weak topologies
$\sigma\bigl(L^1(\X,\mu),L^\infty(\X,\mu)\bigr)$.
\end{lemma}

\begin{proof}
By \cite[Chapter~IV, Theorem~4.3]{SW99}, which holds for arbitrary families
of locally convex spaces, the weak topology of $V$ is the product of the weak
topologies of the spaces $L^1(\X,\mu)$. The claim follows because the topology
induced by a product topology on a product of subsets is the product of the
induced topologies.
\end{proof}

\begin{lemma}\label{lem:F-compact}
    $\F$ is a compact subset of $\bigl(V,\sigma(V,V^*)\bigr)$.
\end{lemma}

\begin{proof}
Fix $\mu\in\Delta(\X)$. Since $0\le h\le 1$ for every $h\in\F'_\mu$ and $\mu$
is a finite measure, the set $\F'_\mu$ is bounded in $L^1(\X,\mu)$ and
uniformly integrable. By the Dunford-Pettis theorem
\cite[Theorem~4.30]{brezis_functional_2011}, $\F'_\mu$ is relatively compact in the weak
topology $\sigma\bigl(L^1(\X,\mu),L^\infty(\X,\mu)\bigr)$.

To complete the proof, we need to show that $\F'_\mu$ is weakly closed. Notice first that it is convex, and it is closed
in the norm topology. To see this, take a converging sequence $h_n\in\F'_\mu$ with $h_n\to h$ in $L^1(\X,\mu)$. Then, by dominated convergence,
a subsequence converges $\mu$-a.e., so that $h$ takes values in
$[0,1]$ and is increasing up to a $\mu$-null set $N$.

Redefining
$h(x)=\sup\{h(y):\ y\in\X\setminus N,\ y\le x\}$ for $x\in N$, with
$\sup\emptyset=0$, yields an increasing function. Since a convex and
norm-closed subset of a Banach space is weakly closed by Mazur's Theorem 
\cite[Theorem~3.7]{brezis_functional_2011}, $\F'_\mu$ is weakly closed, and hence weakly
compact.

By Tychonoff's theorem, $\F=\prod_{\mu\in\Delta(\X)}\F'_\mu$ is compact for
the product of the topologies induced on the sets $\F'_\mu$ by the weak
topologies. By Lemma~\ref{lem:weak-product-F}, this is the topology induced on
$\F$ by $\sigma(V,V^*)$.
\end{proof}
It is straightforward to show that $\F^*$ is also compact and convex.
Convexity follows form the linearity of the constraints. Moreover, since $\F^*$ is closed and $\F^* \subseteq \F$, then $\F^*$ is compact.
This concludes the proof of Theorem \ref{thrm}.
\end{proof}
\subsection{Proof of Proposition \ref{optimalq}.}

 By the definition of product topology, \eqref{eq:functional} is continuous with respect to the weak topology on $L^1(\mu)$ for each $\mu \in \Mprob$. Therefore, we can apply Theorem \ref{bauer}, and the maximizers of \eqref{eq:functional} are of the form of Theorem \ref{thrm}.

Lastly, we show that maximizers of \eqref{eq:functional} are also maximizers of \eqref{payoffexp}.
\begin{lemma}
$q^*$ is independent of $\mathbf{M}$.
\end{lemma}
\begin{proof}
Finally, define $q^*$ as in \eqref{optimalf}, for all  $\mu \in \Mprob$. Then, for any $q \in \mathcal{Q}$,  we have
\begin{align*}
\Pi(q)=\int_{\mathcal{M}_{\text {prob }}} \Pi(q \mid \mu) d \mathrm{M}(\mu)  &\leq \int_{\mathcal{M}_{\text {prob }}} \max _{q' \in \mathcal{Q}} \Pi(q' \mid \mu) d \mathrm{M}(\mu) \\
& =\int_{\mathcal{M}_{\text {prob }}} \Pi\left(q^* \mid \mu\right) d \mathrm{M}(\mu)=\Pi\left(q^*\right) \ .
\end{align*}
\end{proof}

\subsection{Proof of Theorem \ref{thm:endogenous}}
\begin{proof} Several steps are involved.
\setcounter{step}{0}
\begin{step}
    $\F^W$ is compact and convex.
\end{step}
\begin{proof}
    $\F^W$ is nonempty since $(kV, \delta_{\underline{x}}) \in \F^W$. Define $B(W, \mu):=\int g_W d \mu$ on $[0,V] \times \Delta(X)$, and let $(W_n, \mu_n) \in \F^W$ converge to $(W, \mu) \in [0, V] \times \Delta(X)$. Since $g_W$ is bounded,  by triangle inequality we have that
    \begin{align*}
        \left| \int g_{W_n} d\mu_n - \int g_W d \mu \right| &\leq ||g_{W_n} - g_W ||_{\infty} + \left| \int g_{W} d\mu_n - \int g_W d \mu \right| \\
        &=\frac{|W_n - W|}{V} + \left| \int g_{W} d\mu_n - \int g_W d \mu \right| \longrightarrow 0 \ ;
    \end{align*}
    hence, $B(W, \mu)$ is continuous. Since $\{B=k\}$ is the preimage of a closed set under a continuous function, it is closed.

    Now, define
    \begin{equation*}
        d_W(x):=\max\{c(x)-(V-W),0 \}.
    \end{equation*}
    It holds that
    \begin{equation*}
        \mu(X_W)=1 \iff \int d_W d \mu=0 \ .
    \end{equation*}
    By the same reasoning as above, we have that
    \begin{align*}
        \left| \int d_{W_n} d\mu_n - \int d_W d\mu \right| \leq |W_n - W | + \left| \int d_W d\mu_n - \int d_W d\mu \right| \longrightarrow 0 \ .
    \end{align*}
    Hence, since $\int d_{W_{n}} d\mu_n=0$ for every $n$, the limit $\int d_W d\mu=0$, and hence $\mu(X_W)=1$. Therefore, $\F^W$ is the intersection of two closed sets, so it closed. Lastly, since $[0,V] \times \Delta(X)$ is the product of compacts sets, it is also compact by Tychonoff's theorem. As $\F^W \subseteq [0,V] \times \Delta(X)$, it is also compact.
\end{proof}
\begin{step}
    The objective function $J$ is continuous.
\end{step}
\begin{proof}
    Since $\Phi$ is continuous on the compact set $[0,V] \times \Delta(X)$, it is uniformly continuous. Combined with weak convergence, it yields
    \begin{equation*}
        |J(W_n, \mu_n)- J(W, \mu)| \leq \sup_{x \in \X} | \Phi(g_{W_n},x)- \Phi(g_W, x) | + \left| \int \Phi(g_{W},x) d\mu_n-\int \Phi(g_W, x) d\mu \right| \longrightarrow 0 \ ;
    \end{equation*}
    hence, the functional is continuous.
\end{proof}
By Weierstrass' theorem, $\max_{\F^W} J$ exists.
\begin{lemma}
    Let $W \in [0, V]$ and define $P_W:= \{\mu \in \Delta(X): \ (W, \mu) \in \F^W \}$. Then, if $P_W$ nonempty, $J(W, \mu)$ has a maximum on $P_W$ and it is attained by a measure with at most two atoms.
\end{lemma}
\begin{proof}
For fixed $W$, $P_W$ is convex since it is the intersection of the convex set $\Delta(X)$ with two linear constraints. Moreover, $P_W$ is closed in the weak topology: if $\mu_n \in P_W$ and
$\mu_n \to \mu$, then $\mu(X_W) \geq \limsup_n\mu_n(X_W)=1$ because $X_W$ is
closed, and $\int g_W\,d\mu=\lim_n\int g_W\,d\mu_n=k$ because $g_W$ is
bounded and continuous. Being a closed subset of the compact set
$\Delta(X)$, $P_W$ is compact.
\end{proof}
Hence, $P_W$ is a compact and convex subset of $\Delta(X_W)$ subject to a single moment constraint, and $J(W, \cdot)$ is linear and weakly continuous, by Bauer's maximum principle a maximizer can be found among the extreme points of the feasible set, which, by \cite{dubins_extreme_1962} and \cite{winkler_extreme_1988} are convex combinations of at most two Dirac measures.

This shows $(i)$. In order to show $(ii)$, define $q^*(\cdot, \mu^*)$ as in the statement of the theorem. Clearly, since $y_1 \leq y_2$ in $X_{W^*}$, $q^*(\cdot, \mu^*)$ is non-decreasing with values in $[0,1]$. It holds that $q^*(\cdot, \mu^*) \leq g_{W^*}$ on $X$, with equality at $y_1, y_2$. We then have the following result.
\begin{lemma}[Implementation]\label{lem:implementation}
Let $W\in[0,V]$ and $\mu\in\Delta(X)$ satisfy $\mu(X_W)=1$ and
$\int_X g_W\,d\mu=k$, and let $g:X\to[0,1]$ be non-decreasing with
\begin{equation}\label{eq:impl}
g\le g_W \ \text{ on } \X
\qquad\text{and}\qquad
g=g_W \ \text{ on } \operatorname{supp}\mu .
\end{equation}
Then every allocation rule $q$ with $q(\cdot,\mu)=g$ admits $\mu$ as an
equilibrium, every agent's equilibrium payoff equals $W$, and
$\Pi(q,\mu)=J(W,\mu)$.
\end{lemma}

\begin{proof}
First, $q(\cdot,\mu)=g$ is admissible: $g$ is non-decreasing with values
in $[0,1]$, and $\int_X g\,d\mu=\int_X g_W\,d\mu=k$ because $g=g_W$
$\mu$-almost everywhere. Next, for every $x\in \X$,
\[
U_g(x)=g(x)V-c(x)\le g_W(x)V-c(x)=W ,
\]
with equality whenever $x\in\operatorname{supp}\mu$. Since
$\operatorname{supp}\mu\neq\emptyset$, it follows that
$\sup_X U_g=W$ and $\operatorname{supp}\mu\subseteq A(g)$, hence
$\mu(A(g))=1$: $\mu$ is an equilibrium of $q$, and every agent's payoff
is $W$. Finally,
\[
\Pi(q,\mu)=\int_X\Phi\big(g(x),x\big)\,d\mu(x)
=\int_X\Phi\big(g_W(x),x\big)\,d\mu(x)=J(W,\mu),
\]
again because $g=g_W$ $\mu$-almost everywhere.
\end{proof}
Therefore, \eqref{eq:impl} holds, and
Lemma~\ref{lem:implementation} yields that $\mu^*$ is an equilibrium of
$q^*$ with payoff $W^*$ for every agent, which is (iii), and that
$\Pi(q^*,\mu^*)=J(W^*,\mu^*)=\max_{\mathcal F^W}J=\Pi^*$, so the pair is
optimal.
\end{proof}

\subsection{Proof of Theorem \ref{thrm:endogenous2}}
\begin{proof}
The proof is divided in three steps.
\setcounter{step}{0}
\begin{step} \label{step1}
    Every equilibrium payoff is bounded by $k\pi(x_m)$.
\end{step}
Let $q\in\mathcal Q$ and let $\mu$ be an equilibrium of $q$ with rent
$W$, and write $g:=q(\cdot,\mu)$. By Lemma \ref{lem:implementation},
$g=g_W$ $\mu$-a.e., $W\ge0$, and $\mu(X_W)=1$. Since $X_W\subseteq \X_0$
for $W\ge0$, we have $\pi\le\pi(x_m)$ $\mu$-a.e. Moreover $g\ge0$, so
\[
\Pi(q,\mu)=\int_X g(x)\pi(x)\,d\mu(x)
\le\pi(x_m)\int_X g(x)\,d\mu(x)=k\,\pi(x_m),
\]
where the last equality follows from the budget constraint. Taking the supremum
over all pairs yields $\Pi^*\le k\,\pi(x_m)$.
\begin{step}
    The equilibrium pair is feasible.
\end{step}
Set $W^*:=\alpha V-c(x_m)=\max\{\,kV-c(x_m),\,0\,\}\ge0$ and
$\lambda:=k/\alpha\in(0,1]$. We verify the hypotheses of
Lemma~\ref{lem:implementation} for $(W^*,\mu^*)$ and
$g:=\alpha\,\mathbf 1_{[x_m,\bar x]}$.

\emph{Support.} $g_{W^*}(x_m)=(c(x_m)+W^*)/V=\alpha\le1$, so $x_m\in
X_{W^*}$; and $g_{W^*}(\underline x)=W^*/V\le\alpha\le1$, so
$\underline x\in \X_{W^*}$. Hence $\mu^*(X_{W^*})=1$.

\emph{Budget.} $\int g_{W^*}\,d\mu^*=\lambda\,\alpha+(1-\lambda)\,W^*/V$.
If $c(x_m)\le kV$, then $\alpha=k$, $\lambda=1$, and the integral equals $k$.

\emph{Condition \eqref{eq:impl}.} On $\operatorname{supp}\mu^*$: at $x_m$,
$g(x_m)=\alpha=g_{W^*}(x_m)$; at $\underline x$ (relevant only when
$\lambda<1$, i.e.\ $W^*=0$), $g(\underline x)=0=g_{W^*}(\underline x)$.
Off the support: for $x<x_m$, $g(x)=0\le g_{W^*}(x)$; for $x\ge x_m$,
$g(x)=\alpha=g_{W^*}(x_m)\le g_{W^*}(x)$ because $g_{W^*}$ is
non-decreasing. Thus $g\le g_{W^*}$ on $X$ with equality on the support.

By Lemma~\ref{lem:implementation}, $\mu^*$ is an equilibrium of every
$q^*\in\mathcal Q$ with $q^*(\cdot,\mu^*)=g$, with rent $W^*$, and
$\Pi(q^*,\mu^*)=J(W^*,\mu^*)$.
\begin{step}
    The payoff is attained.
\end{step}
\[
J(W^*,\mu^*)=\lambda\,g_{W^*}(x_m)\,\pi(x_m)+(1-\lambda)\,g_{W^*}(\underline x)\,\pi(\underline x)
=\lambda\alpha\,\pi(x_m)+(1-\lambda)\frac{W^*}{V}\pi(\underline x).
\]
In the case $\lambda=1$ the second term vanishes; in the case $\lambda<1$
we have $W^*=0$ and it vanishes again. Hence
$J(W^*,\mu^*)=\lambda\alpha\,\pi(x_m)=k\,\pi(x_m)$. Combined with Step \ref{step1},
$\Pi^*=k\,\pi(x_m)$, attained by $(q^*,\mu^*)$. This proves (i) and (ii).

\end{proof}

\subsection{Proof of Proposition \ref{prop:continuous}}
\begin{proof}
Let $q\in\mathcal Q$ and $\mu$ an equilibrium of
$q$ with rent $W$. By Lemma \ref{lem:implementation}, $W \geq 0$,
$\mu(X_W)=1$, and $q(\cdot,\mu)=g_W$ $\mu$-a.e.; integrating the latter
against $\mu$ and using the budget gives the \emph{rent identity}
\begin{equation}\label{eq:rent-identity}
\int_X c\,d\mu \;=\; kV-W \;\le\; kV .
\end{equation}
Since $\pi'=\psi\circ c$, the designer's payoff is
$\Pi(q,\mu)=\int_X\psi(c(x))\,d\mu(x)$, whose mean is at most $kV$ by
\eqref{eq:rent-identity}.

\setcounter{step}{0}
\begin{step}
 $(i)$ $\psi$ concave.   
\end{step}
\emph{Upper bound.} By
Jensen's inequality and monotonicity of $\psi$,
\[
\Pi(q,\mu)=\int_X\psi(c)\,d\mu\;\le\;\psi\Big(\int_X c\,d\mu\Big)
=\psi(kV-W)\;\le\;\psi(kV)=\pi'(y^*).
\]

Take the equilibrium values from the statement of the theorem. Then $g_0(y^*)=c(y^*)/V=k\le1$, so
$y^*\in \X_0$ and $\mu^*(X_0)=1$. Moreover, $\int g_0\,d\mu^*=k$, $g(y^*)=k=g_0(y^*)$,
and $g\le g_0$ on $X$ because $g=0$ below $y^*$ and $g=k=g_0(y^*)\le
g_0(x)$ above. By Lemma~\ref{lem:implementation}, $\mu^*$ is an
equilibrium of every $q^*$ with $q^*(\cdot,\mu^*)=g$, with rent $0$, and
$\Pi(q^*,\mu^*)=\pi'(y^*)$. Hence $\Pi^*=\pi'(y^*)$.

To show uniqueness, notice that,
if $\psi$ is strictly concave
and $(q,\mu)$ is optimal, both inequalities above are equalities; the
first forces $c$ to be $\mu$-a.s.\ constant, hence $\mu=\delta_y$ for
some $y$ (as $c$ is strictly increasing), and if $\pi'$ is strictly
increasing the second forces $W=0$, so $c(y)=kV$ and $y=y^*$.

\begin{step}
  \textit{(ii)} $\psi$ convex.  
\end{step}
Let
$t:=\int c\,d\mu\le kV$. Every $x\in \X_W\subseteq \X_0$ has
$c(x)\in[0,V]$, so $\mu$ is a distribution on costs supported in $[0,V]$
with mean $t$. Among such distributions a convex $\psi$ is maximized by
the two-point distribution on $\{0,V\}$ with the same mean :
\[
\Pi(q,\mu)=\int_X\psi(c)\,d\mu\;\le\;\Big(1-\frac tV\Big)\psi(0)+\frac tV\psi(V)
\;\le\;(1-k)\,\psi(0)+k\,\psi(V),
\]
where the last step uses $t\le kV$, $\psi(V)\ge\psi(0)$, and the chord inequality. Since
$\psi(0)=\pi'(\underline x)$ and $\psi(V)=\pi'(\overline{x})$, this is
$\Pi(q,\mu)\le(1-k)\pi'(\underline x)+k\,\pi'(\overline{x})$.

\emph{Attainment.} As above, consider the equilibrium values from the statement of the theorem. Then $g_0(\underline x)=0$ and
$g_0(\overline{x})=1$, so both atoms lie in $X_0$. Moreover, we have that
$\int g_0\,d\mu^*=k$, and $g$ agrees with $g_0$ at both atoms. Also, 
$g\le g_0$ on $X$ since $g=0$ below $\overline{x}$ and
$g=1=g_0(\overline{x})\le g_0(x)$ above. By
Lemma~\ref{lem:implementation}, $\mu^*$ is an equilibrium with rent $0$
and $\Pi(q^*,\mu^*)=(1-k)\pi'(\underline x)+k\pi'(\overline{x})$, so the
bound is attained.

In both cases the attaining pair has $W^*=0$.
\end{proof}

\subsection{Proof of Theorem \ref{thrm:MTV}}
It is instructive to derive the following four lemmas, before turning to the proof of Theorem \ref{thrm:MTV}. All results take key insights from \cite{morgan_limits_2022} and translate them in the current setting.

\begin{lemma}[`Housekeeping' lemma]\label{lem:logcost}
Under Assumption~\ref{ass:cost}, the function $c(x)=C(e^{x})$ is twice
continuously differentiable on $\mathbb R$ with $c'>0$ and $c''>0$;
$c'(x)\to0$ as $x\to-\infty$ and $c'(x)\to\infty$ as $x\to\infty$; and
\begin{equation}\label{eq:logconvex}
c'^{\,2}\le c\,c''\qquad\text{and}\qquad (\log c')'\ \text{is non-decreasing.}
\end{equation}
\end{lemma}

\begin{proof}
Denote $q=e^{x}$. Then $c'(x)=qC'(q)>0$ and
$c''(x)=qC'(q)+q^{2}C''(q)=c'(x)\big(1+\eta_{C'}(q)\big)>0$, so that
$(\log c')'(x)=1+\eta_{C'}(q)$ is non-decreasing because $\eta_{C'}$ is.
The limits follow from $C'(0^{+})<\infty$ and $qC'(q)\to\infty$.

For the
first part of \eqref{eq:logconvex}, since $sC'(s)\to0$ as $s\downarrow0$,
\[
qC'(q)=\int_{0}^{q}\frac{d}{ds}\big[sC'(s)\big]\,ds
=\int_{0}^{q}C'(s)\big(1+\eta_{C'}(s)\big)\,ds
\le\big(1+\eta_{C'}(q)\big)C(q),
\]
using that $\eta_{C'}$ is non-decreasing. Hence
$c'/c=qC'(q)/C(q)\le1+\eta_{C'}(q)=c''/c'$, which implies  $c'^{\,2}\le cc''$.
\end{proof}

\begin{lemma}[Concavity of $\kappa$]\label{lem:cap}
Let Assumption~\ref{ass:density} hold and let
$\kappa(\phi):=f\big(\bar F^{-1}(\phi)\big)$ for $\phi\in(0,1)$, with
$\kappa(0):=\kappa(1):=0$. Then $\kappa$ is continuous on $[0,1]$ and
strictly concave on $(0,1)$. Moreover, for every rule $g$ and every
effort $x\in\mathbb R$,
\begin{equation}\label{eq:capacity}
P_g'(x)\;\le\;\kappa\big(P_g(x)\big),
\end{equation}
and if equality holds at some $x$ with $P_g(x)\in(0,1)$, then
$g=\mathbb{I}_{[t,\infty)}$ with $t=x+\bar F^{-1}(P_g(x))$.
\end{lemma}

\begin{proof}
\emph{Concavity.} On $(0,1)$ the quantile function $\bar F^{-1}$ is
differentiable with derivative $-1/f(\bar F^{-1}(\phi))$, so
\[
\kappa'(\phi)=f'\big(\bar F^{-1}(\phi)\big)\cdot\Big(-\frac{1}{f(\bar F^{-1}(\phi))}\Big)
=-(\log f)'\big(\bar F^{-1}(\phi)\big).
\]
Under strict log-concavity $(\log f)'$ is strictly decreasing, and
$\bar F^{-1}$ is strictly decreasing; hence $\kappa'$ is strictly
decreasing and $\kappa$ is strictly concave. Continuity at the endpoints
holds because $f(z)\to0$ as $z\to\pm\infty$.

\emph{Random-cutoff representation.} A rule $g$ is non-decreasing with
values in $[0,1]$; hence, it is the distribution function of a random variable
$T$ with values in $[-\infty,+\infty]$: $g(y)=\Pr(T\le y)$, with
$\Pr(T=-\infty)=\lim_{y\to-\infty}g(y)$ and
$\Pr(T=+\infty)=1-\lim_{y\to+\infty}g(y)$.

Take $T$ independent of the
noise. An agent with effort $x$ wins if and only if $x+\varepsilon\ge T$,
so, conditioning on $T$ and with the conventions $\bar F(-\infty)=1$,
$\bar F(+\infty)=0$, $f(\pm\infty)=0$, we have that
\[
P_g(x)=\mathbb E\big[\bar F(T-x)\big],\qquad \text{and} \qquad 
P_g'(x)=\mathbb E\big[f(T-x)\big] \ .
\]

\emph{The bound.} Let $\Phi:=\bar F(T-x)\in[0,1]$. By definition of
$\kappa$, $f(T-x)=\kappa(\Phi)$ for every value of $T$, including
$T=\pm\infty$ where both sides vanish. By Jensen's inequality,
\[
P_g'(x)=\mathbb E\big[\kappa(\Phi)\big]\le\kappa\big(\mathbb E[\Phi]\big)=\kappa\big(P_g(x)\big).
\]

\emph{Equality.} If equality holds at some $x$ with
$\mathbb E[\Phi]=P_g(x)\in(0,1)$, strict concavity forces $\Phi$ to be
almost surely constant, equal to $P_g(x)\in(0,1)$. Since
$\bar F(\cdot-x)$ is strictly decreasing and maps $\pm\infty$ to
$\{0,1\}$, $T$ is almost surely the finite constant
$t=x+\bar F^{-1}(P_g(x))$, i.e.\ $g=\mathbb{I}_{[t,\infty)}$.
\end{proof}

\begin{lemma}[Rent floor]\label{lem:floor}
Let Assumptions~\ref{ass:cost} and~\ref{ass:density} hold. Then
$\Lambda$ is continuously differentiable on $(0,1)$ and, for every
$\phi\in(0,1)$,
\begin{equation}\label{eq:soc-equiv}
\Lambda'(\phi)>0\quad\Longleftrightarrow\quad
c''\big(x(\phi)\big)>-v\,f'\big(\bar F^{-1}(\phi)\big).
\end{equation}
Moreover, $\Lambda(\phi)\ge0$ implies $\Lambda'(\phi)>0$. Hence,
if $\Lambda(k)\ge0$, then $\Lambda$ is strictly increasing on $[k,1)$.
\end{lemma}

\begin{proof}
By Lemma~\ref{lem:logcost} and the inverse function theorem, $x(\cdot)$
is continuously differentiable on $(0,1)$ with
$c''(x(\phi))\,x'(\phi)=v\kappa'(\phi)$. Using $c'(x(\phi))=v\kappa(\phi)$,
\[
\Lambda'(\phi)=v-c'(x(\phi))\,x'(\phi)
=v\Big[1-\frac{v\,\kappa(\phi)\,\kappa'(\phi)}{c''(x(\phi))}\Big].
\]
Denoting $z:=\bar F^{-1}(\phi)$,
Lemma~\ref{lem:cap} yields $\kappa(\phi)\kappa'(\phi)=-f'(z)$, and
\eqref{eq:soc-equiv} follows.
\setcounter{step}{0}
\begin{step}\label{steplambda}
   $\Lambda(\phi)\ge0\Rightarrow\Lambda'(\phi)>0$. 
\end{step}
If
$\kappa'(\phi)\le0$ the bracket is at least $1$. Otherwise, using
$v\kappa=c'$ at $x(\phi)$, then $c'^{\,2}\le cc''$ from
\eqref{eq:logconvex}. Moreover, we have that  $c(x(\phi)) \leq v\phi$ (which is
$\Lambda(\phi)\ge0$), which implies
\[
\frac{v\kappa\kappa'}{c''}
=\frac{c'(x(\phi))^{2}}{c''(x(\phi))}\cdot\frac{\kappa'}{v\kappa}
\le\frac{c(x(\phi))}{v}\cdot\frac{\kappa'}{\kappa}
\le\phi\,\frac{\kappa'(\phi)}{\kappa(\phi)}
=\frac{\bar F(z)\,\big(-f'(z)\big)}{f(z)^{2}}=:\rho(\phi).
\]
Since $\frac{d}{dz}\big[\bar F/f\big]=-1-\bar Ff'/f^{2}=\rho(\phi)-1$ and
$\bar F/f=1/h$ is strictly decreasing by Assumption~\ref{ass:density},
we get $\rho(\phi)<1$, hence $\Lambda'(\phi)\ge v\,(1-\rho(\phi))>0$.

\begin{step}
    Monotonicity on $[k,1)$.
\end{step}
Let $\Lambda(k)\ge0$ and suppose
$S:=\{\phi\in[k,1):\Lambda(\phi)<0\}$ is nonempty. Denote $\phi_0:=\inf S$.
By continuity, $\Lambda(\phi_0)\ge0$ (by hypothesis if $\phi_0=k$, as a
left limit otherwise).

By Step \ref{steplambda} and continuity of $\Lambda'$, $\Lambda$
is increasing on a right neighbourhood of $\phi_0$ and hence
non-negative, contradicting $\phi_0=\inf S$. Thus, $\Lambda\ge0$ on
$[k,1)$, and by Step~1 $\Lambda'>0$ on that set.
\end{proof}

\begin{lemma}[Unique maximizer]\label{lem:hump}
Let Assumptions~\ref{ass:cost} and~\ref{ass:density} hold, and let
$g=\alpha\mathbb{I}_{[t,\infty)}$ with $\alpha\in(0,1]$ and $t\in\mathbb R$.
If $x_0\in\mathbb R$ satisfies the first-order condition
$vP_g'(x_0)=c'(x_0)$ and the second-order condition
$vP_g''(x_0)<c''(x_0)$, then $x_0$ is the unique maximizer of
$U_g=vP_g-c$ on $\mathbb R$, and the equilibrium payoff of every agent
under $g$ is $\max\{U_g(x_0),0\}$.

Moreover, the set of best responses is
$\{x_0\}$ if $U_g(x_0)>0$, $\{-\infty,x_0\}$ if $U_g(x_0)=0$, and
$\{-\infty\}$ if $U_g(x_0)<0$.
\end{lemma}

\begin{proof}
On $\mathbb R$, $U_g'(x)=\alpha vf(t-x)-c'(x)=c'(x)\big(e^{D(x)}-1\big)$
with
\[
D(x):=\log\big(\alpha vf(t-x)\big)-\log c'(x).
\]
$D$ is concave and continuously differentiable: $x\mapsto\log f(t-x)$ is
concave by Assumption~\ref{ass:density}, and $-\log c'$ is concave by
\eqref{eq:logconvex}.

At $x_0$,
$U_g'(x_0)=0$ gives $D(x_0)=0$; differentiating $U_g'=c'(e^{D}-1)$ at a
zero of $D$ gives $U_g''(x_0)=c'(x_0)D'(x_0)$, so the second-order
condition becomes $D'(x_0)<0$.

By concavity, it holds that $D<0$ 
and $U_g$ is strictly decreasing on on $(x_0,\infty)$. The superlevel set $\{D>0\}$ is
an interval which contains a left neighbourhood of $x_0$ (since
$D(x_0)=0>D'(x_0)$) and not $x_0$. Denote this interval as $(a,x_0)$ with
$a\in[-\infty,x_0)$.

If $a>-\infty$ then $D(a)=0$, and concavity gives
$D<0$ on $(-\infty,a)$. Hence $U_g$ is strictly decreasing on
$(-\infty,a]$, strictly increasing on $[a,x_0]$ and strictly decreasing
on $[x_0,\infty)$. Moreover, $x_0$ is the unique maximizer of $U_g$
on $\mathbb R$.

Finally, $U_g(x)\to\alpha v\cdot0-c(-\infty)=0$ as $x\to-\infty$, and
$U_g<0$ on $(-\infty,a)$ when $a>-\infty$. Including the drop-out option
$U_g(-\infty)=0$, the supremum over $[-\infty,\infty)$ is
$\max\{U_g(x_0),0\}$, with the stated best-response sets.
\end{proof}

\begin{proof}[Proof of Theorem~\ref{thrm:MTV}]
Let $(g,\mu)$ be an equilibrium pair with rent $W\ge0$, and let
$x_n:=\sup\operatorname{supp}\mu$. Since $U_g(x)\le v-c(x)\to-\infty$ as
$x\to\infty$, $x_n$ is finite.

Moreover, since $\int P_g\,d\mu=k>0=P_g(-\infty)$,
$\mu$ is not concentrated on drop-out. By continuity of $U_g$,
$x_n\in A(g)$; by monotonicity of $P_g$, it follows that $\phi_n:=P_g(x_n)\ge\int P_g\,d\mu=k$ by monotonicity
of $P_g$, while $\phi_n<1$ because $P_g\equiv1$ above $x_n$ would give
$P_g'(x_n)=0\ne c'(x_n)/v$. Recall the following two facts:

\emph{Rent identity.} $vP_g-c=W$ on $\operatorname{supp}\mu$; integrating
against $\mu$ and using the budget, $\int c\,d\mu=kv-W$.

\emph{Jensen's inequality.} $C$ is convex and $C(e^{x})=c(x)$, so
$C(O(\mu))\le\int C(e^{x})\,d\mu=\int c\,d\mu$.

\setcounter{step}{0}
\begin{step}
    Upper bound in \textit{(i)}.
\end{step}
Since $c$ and $c'$ are
strictly increasing,
\begin{equation}\label{eq:regimes}
\Lambda(k)\ge0\iff c(x(k))\le kv=c(x^{\circ})\iff x(k)\le x^{\circ}
\iff e^{x(k)}\le C^{-1}(kv).
\end{equation}

The rent identity, $W\ge0$ and
Jensen's inequality give $C(O(\mu))\le kv$, i.e.\ $O(\mu)\le C^{-1}(kv)$. If
$\Lambda(k)<0$, this is the minimum in (i) by \eqref{eq:regimes}.

Let
$\Lambda(k)\ge0$. At the top atom the first-order condition
$vP_g'(x_n)=c'(x_n)$ and Lemma~\ref{lem:cap} give
$c'(x_n)\le v\kappa(\phi_n)=c'(x(\phi_n))$, so $x_n\le x(\phi_n)$ and
\[
W=v\phi_n-c(x_n)\;\ge\;v\phi_n-c\big(x(\phi_n)\big)=\Lambda(\phi_n)\;\ge\;\Lambda(k),
\]
where the last step follows from  Lemma~\ref{lem:floor} and $\phi_n\ge k$. By the rent
identity $\int c\,d\mu\le kv-\Lambda(k)=c(x(k))$, and by Jensen
$C(O(\mu))\le C(e^{x(k)})$, i.e.\ $O(\mu)\le e^{x(k)}$.

\begin{step}
    \textit{(ii)} $\Lambda(k)\ge0$.
\end{step}
Let
$t^*:=x(k)+\bar F^{-1}(k)$ and $g^*:=\mathbb{I}_{[t^*,\infty)}$, so that
$P_{g^*}(x)=\bar F(t^*-x)$. At $x(k)$ the budget is satisfied with equality,
$P_{g^*}(x(k))=k$; the first-order condition is satisfied,
$vP_{g^*}'(x(k))=vf(\bar F^{-1}(k))=v\kappa(k)=c'(x(k))$; and the
second-order condition $-vf'(\bar F^{-1}(k))<c''(x(k))$ is also satisfied by
\eqref{eq:soc-equiv}, since $\Lambda(k)\ge0$ implies $\Lambda'(k)>0$.

By Lemma~\ref{lem:hump}, $x(k)$ is every agent's unique best response,
with payoff $U_{g^*}(x(k))=kv-c(x(k))=\Lambda(k)\ge0$. Hence
$(g^*,\delta_{x(k)})$ is an equilibrium pair with $O=e^{x(k)}$, and the
bound is attained.

\begin{step}
    Uniqueness.
\end{step}

Let $C''>0$ and let $(g,\mu)$ attain $e^{x(k)}$. Then
every inequality in (i) is an equality. Jensen with equality for the
strictly convex $C$ forces $e^{x}$ to be $\mu$-a.s.\ constant, so
$\mu=\delta_{x_n}$.

Moreover, $W=\Lambda(\phi_n)=\Lambda(k)$ and the strict
monotonicity of $\Lambda$ imply $\phi_n=k$. In addition to this,
$c(x_n)=c(x(\phi_n))$ implies $x_n=x(k)$. Lastly, 
$c'(x_n)=v\kappa(k)=v\kappa(P_g(x_n))$ holds with equality in
Lemma~\ref{lem:cap} at a point with $P_g(x_n)=k\in(0,1)$, so
$g=\mathbb{I}_{[x(k)+\bar F^{-1}(k),\infty)}=g^*$.

\begin{step}
  $(iii) \ \Lambda(k)<0$.  
\end{step}
By
\eqref{eq:regimes}, $v\kappa(k)>c'(x^{\circ})$, i.e.\
$h(\bar F^{-1}(k))=\kappa(k)/k>c'(x^{\circ})/kv$. Since $h$ is
continuous and strictly increasing with $h(z)\to0$ as $z\to-\infty$, the
equation $h(z)=c'(x^{\circ})/kv$ has a unique solution
$z^{\circ}<\bar F^{-1}(k)$.

Set $\alpha:=k/\bar F(z^{\circ})\in(0,1)$,
$t^{\circ}:=x^{\circ}+z^{\circ}$ and $g^{\circ}:=\alpha\mathbb{I}_{[t^{\circ},\infty)}$,
so that $P_{g^{\circ}}(x)=\alpha\bar F(t^{\circ}-x)$. At $x^{\circ}$:
the budget holds with equality, $P_{g^{\circ}}(x^{\circ})=\alpha\bar F(z^{\circ})=k$;
the rent is zero, $vP_{g^{\circ}}(x^{\circ})-c(x^{\circ})=kv-kv=0$; and
the first-order condition holds,
\[
vP_{g^{\circ}}'(x^{\circ})=\alpha vf(z^{\circ})
=\alpha v\bar F(z^{\circ})\,h(z^{\circ})=kv\cdot\frac{c'(x^{\circ})}{kv}=c'(x^{\circ}).
\]
For the second-order condition $-\alpha vf'(z^{\circ}) < c''(x^{\circ})$,
consider the contest with prize $v':=\alpha v$ and prize mass
$k':=\bar F(z^{\circ})$. It holds that $v'\kappa(k')=\alpha vf(z^{\circ})=c'(x^{\circ})$,
so $x^{\circ}$ is the utility-maximizing effort at $k'$ and
$\Lambda_{\alpha}(k')=v'k'-c(x^{\circ})=kv-kv=0$. Lemma~\ref{lem:floor}
 yields the second-order condition at $x^{\circ}$.
 
By Lemma~\ref{lem:hump}, $x^{\circ}$ is the unique best response in
$\mathbb R$, with payoff $0$, and agents are indifferent to dropping out;
so $(g^{\circ},\delta_{x^{\circ}})$ is an equilibrium pair with
$O=e^{x^{\circ}}=C^{-1}(kv)$, and the first best is attained.

\begin{step}
Uniqueness of the equilibrium distribution.    
\end{step}
Let $C''>0$ and let
$(g,\mu)$ attain $C^{-1}(kv)$. Then
$kv=C(O(\mu))\le\int c\,d\mu=kv-W\le kv$, so $W=0$ and Jensen holds with
equality. Therefore $\mu=\delta_{x}$ with $c(x)=kv$, i.e.\
$\mu=\delta_{x^{\circ}}$.

\end{proof}

\subsection{Proof of Corollary \ref{cor:sigma}}
\begin{proof}
\emph{Scaling.} Write $\bar F_\sigma$, $f_\sigma$, $\kappa_\sigma$,
$\Lambda_\sigma$ for the objects of Theorem~\ref{thrm:MTV} computed from
$\varepsilon=\sigma\tilde\varepsilon$, and $\bar F$, $f$, $\kappa$ for
those of $\tilde\varepsilon$. Then $\bar F_\sigma(z)=\bar F(z/\sigma)$ and
$f_\sigma(z)=f(z/\sigma)/\sigma$, so $\bar F_\sigma^{-1}(\phi)=\sigma\bar F^{-1}(\phi)$
and
\[
\kappa_\sigma(\phi)=f_\sigma\big(\bar F_\sigma^{-1}(\phi)\big)
=\frac{1}{\sigma}f\big(\bar F^{-1}(\phi)\big)=\frac{\kappa(\phi)}{\sigma}.
\]
Assumption~\ref{ass:density} is invariant under scaling, so
Theorem~\ref{thrm:MTV} applies for every $\sigma>0$, with
$x(k;\sigma)=(c')^{-1}(v\kappa(k)/\sigma)$ and
$\Lambda_\sigma(k)=kv-c(x(k;\sigma))$.

\setcounter{step}{0}
\begin{step}
 The case distinction.   
\end{step}
Since $c'$ is strictly increasing,
$x(k;\sigma)$ is strictly decreasing in $\sigma$, and by
\eqref{eq:regimes}
\[
\Lambda_\sigma(k)\ge0
\iff x(k;\sigma)\le x^{\circ}
\iff \frac{v\kappa(k)}{\sigma}\le c'(x^{\circ})
\iff \sigma\ge\underline\sigma .
\]

\begin{step}
    Optimality.
\end{step}
If $\sigma\ge\underline\sigma$, Theorem~\ref{thrm:MTV}(ii)
applies: the threshold is optimal, and under $C''>0$ it is the unique
optimal rule, so the lottery with $\alpha<1$ is not optimal.

If
$\sigma<\underline\sigma$, Theorem~\ref{thrm:MTV}(iii) applies: the
lottery is optimal, and no deterministic threshold is. This proves both
equivalences. (At $\sigma=\underline\sigma$ the two rules coincide:
$\Lambda_\sigma(k)=0$ gives $x(k;\sigma)=x^{\circ}$, $z^{\circ}=\bar F_\sigma^{-1}(k)$
and $\alpha=1$.)

\begin{step}
 Output. 
\end{step}
By Theorem~\ref{thrm:MTV}(i)--(iii),
$O^*(\sigma)=\min\{C^{-1}(kv),e^{x(k;\sigma)}\}$. As $x(k;\sigma)$ is
decreasing in $\sigma$, so is $O^*(\sigma)$, and by the regime boundary
$O^*(\sigma)=C^{-1}(kv)$ exactly when $\sigma\le\underline\sigma$.
\end{proof}

\bigskip

\bibliographystyle{alpha}
\bibliography{ExtremePoints}

\pagebreak

\end{document}